\documentclass[11pt]{amsart}

\usepackage[parfill]{parskip}
\usepackage{amssymb}
\usepackage{amsmath}
\usepackage{amsthm}
\usepackage{amsfonts}
\usepackage[T1]{fontenc}
\usepackage{palatino}
\usepackage[text={6.5in, 9in}, centering]{geometry}

\usepackage[usenames]{color}
\newtheorem{theorem}{Theorem}

\newtheorem{corollary}[theorem]{Corollary}

\newtheorem{definition}[theorem]{Definition}

\newcommand{\R}{\mathbb{R}}

\newcommand{\supp}{\operatorname{supp}}

\title[SQG temperature variance cascade]{A note on the surface quasi-geostrophic temperature variance cascade}
\author{Z. Bradshaw}\author{Z. Gruji\'c}
\begin{document}
\begin{abstract}
In this note we examine the dynamical role played by inertial forces on the surface temperature (or buoyancy) variance in strongly rotating, stratified flows with uniform potential vorticity fields and fractional dissipation.  In particular, using a dynamic, multi-scale averaging process, we identify a sufficient condition for the existence of a direct temperature variance cascade across an inertial range. While the result is consistent with the physical and numerical theories of SQG turbulence, the condition triggering the cascade is more exotic, a fact reflecting the non-locality introduced by fractional dissipation.  A comment regarding the scale-locality of the temperature variance flux is also included.
\end{abstract}
\maketitle
\section{Introduction}

The present note studies the dynamical influence of inertial forces on the temperature variance of solutions to the surface quasi-geostrophic (SQG) model and its generalizations. In particular, we are interested in the inertial effect of a strongly rotating, stratified fluid, $u$, on a surface temperature (or buoyancy), $\theta$, where the fluid's streamfunction, $\Psi$,  is independent of the vertical coordinate.  The evolution of $\theta$ is governed by the SQG equations, which have the form,
\begin{align}
\begin{cases}\label{criticalSQG}\big(\partial_t +u\cdot \nabla +\kappa\Lambda^\alpha \big)\theta = 0 		& \mbox{~in~}\R^2\times (0,\infty)
\\u=(-\partial_{2}\Psi, \partial_1 \Psi ),~\nabla\cdot u=0 	& \mbox{~in~}\R^2\times (0,\infty)
%\\\Lambda \Psi = -\theta 	& \mbox{~on~}\R^2\times (0,\infty)
\\ \theta(\cdot, 0)=\theta_0 &\mbox{~in~}\R^2,
\end{cases}
\end{align} where $0<\alpha\leq 2$ and $\kappa\geq 0$ are parameters, the streamfunction and $\theta$ are related by $ \Lambda \Psi = -\theta$, and $\Lambda^\alpha$ is defined as a pseudo-differential operator by, \[ ( \Lambda^\alpha f)^\wedge(\xi)=|\xi|^\alpha f^\wedge (\xi),\]for $\xi\in \R^2$.  Note that $\alpha=1$ corresponds to the \emph{critical} case, while $\alpha<1$ and $\alpha>1$ correspond to the supercritical and subcritical cases respectively. 
For brevity we set $\kappa=1$. 

The SQG equations were introduced in the context of the quasi-geostrophic equations (QG) -- a geophysically relevant 3D model for the displacement of a rotating, stratified fluid from a solid body rotation -- to rectify the fact that the QG description of atmospheric and oceanic processes breaks down in regions near the fluid's boundaries such as the troposphere and ocean surface.  Essentially, there are various dynamical regimes in the fluid media which demand more nuanced descriptions than that given by QG.  For surface proximate flow, this is achieved by assuming the potential vorticity of a QG fluid is uniform.  This assumption allows one to derive the SQG equations from QG.
These models are also mathematically interesting as, due to the fact that strong stratification makes the systems behave in a quasi-lower dimensional fashion, they provide a potentially more friendly setting (compared to the 3D Navier-Stokes equations or 3D Euler equations) to study the problem of global regularity and, additionally, they serve as a testing ground for turbulence theories (cf. \cite{Co_QG_02,CoMaTa_QG_94,Held-et-al}).

A spectral theory of SQG turbulence was initially studied by Blumen in \cite{Blumen1} where he derives Kolmogorov-Kraichnan type scaling laws and associated inertial ranges for the inviscid regime (see also \cite{Held-et-al}).  In SQG there are two ideally conserved quantities, the surface temperature variance and the depth-integrated total energy -- the densities for these quantities are $|\theta|^2$ and $|u|^2$ -- which behave in an analogous fashion to 2D fluid enstrophy and energy in that they exhibit direct and inverse cascades respectively. Recent numerical work in \cite{CoEtAt12_SQG} supports the picture of intermittently distributed, thinning temperature assemblages for $\alpha\in [0,2]$.  Mathematically, the inverse energy cascade in QG turbulence has been studied from a Littlewood-Paley perspective (cf. \cite{Co_QG_02}).
The main goal of this note is to present a rigorous affirmation of the existence of the temperature variance cascade and the scale-locality of the temperature variance flux across an inertial range under conditions which are consistent with qualitative properties of SQG turbulence.

The mathematical methodology used is an adaptation of one developed by Dascaliuc and Gruji\'c to study cascades and flux locality in hydrodynamic turbulence \cite{DaGr1,DaGr5} and extended to study related properties in plasma turbulence \cite{BrGr2}. The adaptation is necessary because fractional dissipation complicates our ability to study dynamical properties of the flow via a localized temperature variance equation (or inequality) as, on $\R^2$, no dominant, \emph{a priori} positive dissipative quantity is immediately apparent.  To overcome this, we incorporate an operator which extends scalar functions of $\R^2$ to $\R^2\times (0,\infty)$.  The fractional Laplacian of a function on $\R^2$ is then recovered via the trace of the normal derivative of the extension.  This approach has been applied to obtain regularity results for critical QG in \cite{CaVa_QG} and certain supercritical versions of QG in \cite{CoWu_QG_09}.

\section{Methodology}
In this section we describe the ensemble averaging process which constitutes the framework on which turbulent cascades will be studied.
Our methodology is an adaptation of the statistical apparatus of Dascaliuc and Gruji\'c originally developed to study hydrodynamic turbulence (cf. \cite{DaGr1}).  We work on a 2D macro-scale domain taken to be $B(x_0,R_0)$, the open ball of radius $R_0$ centered at $x_0$. Multi-scale, statistical comparisons will be achieved using $(K_1,K_2)$-covers which are presently defined.

\begin{definition}Let $K_1,K_2\in \mathbb N$ and $0< R\leq R_0$.  The cover of the macro-scale domain $B(x_0,R_0)$ by the $n$ (open) balls $\{B(x_i,R)\}_{i=1}^n$ is a  {\em $(K_1,K_2)$-cover at scale }$R$ if,
\begin{align*}\bigg( \frac {R_0} R \bigg)^2\leq n \leq K_1 \bigg(\frac {R_0} R\bigg)^2\end{align*}
and, for any $x\in B(x_0,R_0)$, $x$ is contained in at most $K_2$ cover elements.
\end{definition}

In what follows, all covers are understood to be $(K_1,K_2)$-covers.
The positive integers $K_1$ and $K_2$ represent the maximum allowed \emph{global} and \emph{local multiplicities}.

Localization is achieved via refined cut-off functions.  For a cover element centered at $x_i$ of scale $R$, let $\phi_i(x,t)=\eta(t)\psi(x)$ where $\eta\in C_0^\infty(0,\infty)$ and $\psi\in C_0^\infty (\R^2)$ satisfy,
\begin{align}\label{timecutoff}0\leq \eta\leq 1,\qquad \eta=0~\mbox{on }(0,T/3)\cup (5T/3,\infty),\qquad\eta=1~\mbox{ on }[2T/3,4T/3],\qquad\frac {|\partial_t\eta|} {\eta^{2\delta-1} }\leq \frac {C_0} T,
\end{align}
and,
\begin{align}\label{spacecutoff} 0\leq \psi\leq 1,\qquad\psi=1\mbox{ on }B(x_i,R),\qquad\frac {|\partial_i\psi|} {\psi^{\delta}} \leq \frac {C_0} {R},\qquad\frac {|\partial_i\partial_j \psi|} {\psi^{2\delta-1}}\leq \frac {C_0} {R^2},
\end{align}where $1/2< \delta <1$ is fixed.  We require also that $2T\geq R_0^\alpha $ where $\alpha$ reflects the fractional exponent in \eqref{criticalSQG}.  

The macro-scale cut-off function, $\phi_0$, is a fixed, refined cut-off function for the ball $B(x_0,R_0)$ which is supported on $B(x_0,2R_0)$ having spatial and temporal components $\psi_0$ and $\eta_0$.  

Comparisons will be necessary between sub-macro-scale quantities and macro-scale quantities.  To accommodate this we impose several additional conditions on $\phi_i$ when $x_i$ lies near the boundary of $B(x_0,R_0)$.  If $B(x_i,R)\subset B(x_0,R_0)$ we assume $\psi\leq \psi_0$.  Alternatively, when $B(x_i,R)\not\subset B(x_0,R_0)$, let $l(x,y)$ denote the collection of points on the line through $x$ and $y$ and define the sets,  
\begin{align*}
S_0&=B(x_i,R)\cap B(x_0,R_0),
\\S_1&=\big\{x: R_0\leq |x|< 2R_0 ~\mbox{and}~\emptyset\neq \big( l(x,x_0)\cap  \partial B(x_i,R)\cap B(x_0,R_0)^c \big)\big\},
\\S_2&=\bigg( B(x_i,2R)\cup \big\{x: R_0\leq |x|< 2R_0 ~\mbox{and}~\emptyset\neq \big( l(x,x_0)\cap  \partial B(x_i,2R)\cap B(x_0,R_0)^c \big)\big\}\bigg)
\\&\qquad \cap (S_0\cup S_1)^c.
\end{align*}Then, our assumptions are that $\psi$ satisfies (\ref{spacecutoff}), $\psi=1$ on $S_0$, $\psi=\psi_0$ on $S_1$, and $\supp \psi = S_2$.  The above conditions ensure that $\phi\leq \phi_0$ and that $\psi$ can be constructed to have an inwardly oriented gradient field.

Extensions into a third spatial dimension will be necessary when deriving a localized temperature variance equality.  To accommodate this, fix a vertical scale $R^*$ and let $\psi^*\in C^\infty([0,\infty))$ take the value $1$ on $[0,R^*]$ and decrease to $0$ across $[R^*,R^*+R_0)$, in a fashion mirroring the (spatial) gradient restrictions imposed on $\phi_0$.  Then, the vertical extensions of our refined cut-off functions are just $\phi_0^*(x,z,t)=\phi_0(x,t)\psi^*(z)$ and $\phi_i^*(x,z,t)=\phi_i(x,t)\psi^*(z)$.  Note that the property $\phi_i^*\leq \phi_0^*$ is trivially preserved.

The statistical element of the methodology is carefully described in \cite{DaGr1}, although some additional remarks are necessary for our present work.
Let $g$ be a physical density and denote its localized surface-time average over a cover element at scale $R$ around $x_i$ by, \begin{align*} G_i =   \frac 1 T \int_0^{2T} \frac 1 {R^2} \int_{B(x_i,2R)} g(x,t)\phi_{i}(x,t)~dx ,\end{align*}
and let $\langle G\rangle_R$ denote the arithmetic mean -- our version of \emph{ensemble averages} -- over localized averages associated with cover elements,
\[\langle G \rangle_R=\frac 1 n \sum_{i=1}^n   G_i.\]
Assume that $g^*$ is an extension of $g$ in the $z$-direction. Analogous quantities involving the vertical extensions of cover elements are,
\begin{align*} G_i^* =   \frac 1 T \int_0^{2T} \frac 1 {R^2} \int_{B(x_i,2R)} \bigg[\int_0^{R^*+R_0} g^*(x,z,t){\phi^*_i}(x,z,t)~dz\bigg]~dx~dt .\end{align*}
Note that no normalization is carried out in the $z$-direction. This is consistent with the fact that the extended cut-off functions all have the same vertical scale.

We extract statistical information about $g$ at the scale $R$ by studying the collection of all ensemble averages at the given scale.  In particular, for an \emph{a priori} sign varying density, if all covers taken at scale $R$ admit ensemble averages which are mutually comparable and positive, e.g., \[\frac 1 C K \leq \langle G \rangle_R \leq CK,\] for positive, cover independent constants $C$ and $K$, then this density is \emph{essentially positive} (i.e. positive in a statistical sense) at scale $R$.

\section{An equation for the localized temperature variance flux}

Recall that the temperature variance flux through the boundary of a region $B$ is,
\begin{align*}\int_{\partial B} \frac 1 2 \, \theta^2\, u\cdot \hat n \,ds=\int_B (u\cdot \nabla \theta ) \,\theta \,dx,
\end{align*}where $\hat n$ is the outward normal vector.  The refined cut-off functions were constructed so that the gradient field $\nabla \phi_i$ was oriented roughly toward the center of the ball $B(x_i,R)$. This leads to the following realization of the localized temperature variance flux \emph{into} $B(x_i,R)$,
\begin{align*}\int_{\R^2}\frac 1 2 \,\theta^2\, (\nabla \phi_i\cdot u)\,dx=-\int_{\R^2} (u\cdot \nabla \theta)\,\theta \,\phi_i\,dx,
\end{align*}
which is a consequence of the fact that $u$ is divergence free.  An initial formula for the surface-time integrated average flux into $B(x_i,R)$ now follows from \eqref{criticalSQG},
\begin{align*}F_i&:=\int_0^{2T} \int_{\R^2}\frac 1 2 \,\theta^2\, (\nabla \phi_i\cdot u)\,dx\,dt
\\&= \int_0^{2T}\int \phi_i \,\theta\Lambda^\alpha \theta \,dx\,dt - \int_0^{2T}\int  \frac 1 2\, {\theta^2} \, (\partial_t \phi_i)\,dx\,dt.\end{align*}

The methodology of Dascaliuc and Gruji\'c (cf. \cite{DaGr1}) depends crucially on localized energy estimates which are available in the context of suitable weak solutions for 3D NSE (cf. \cite{CKN-82}).  Considering the non-locality introduced by $\Lambda^\alpha$, these local estimates are not directly accessible for solutions to SQG when $\alpha<2$, and, consequently, an indirect approach using an extension operator into a third spatial dimension is necessary. Properties of the extension are detailed in \cite{CaSi}; it has previously been applied to the critical and supercritical formulations of SQG in \cite{CaVa_QG} and \cite{CoWu_QG_09}.  

The desired extension operator, $L$, is defined in \cite{CaSi} via the initial value problem,
\[
\begin{cases}
\nabla\cdot(z^b\nabla L(f))(x,z)=0	&\mbox{~for~} (x,z)\in \R^2\times (0,\infty),
\\L(f)(x,0)=f(x) &\mbox{~for~}x\in \R^2,
\end{cases}
\]where $b=1-\alpha$.
The solution $L(f)$ can be represented with a Poisson-like formula and satisfies, 
\[\Lambda^\alpha f=\lim_{z\to 0}(-z^b \partial_z(Lf)).\] 
In the critical case this reduces to,
\begin{align*}
(\Lambda f )(x)=\partial_\nu (Lf)(x) &\mbox{~for~} x\in \R^2.
\end{align*}

Let $\theta$ and $u$ be solutions to \eqref{criticalSQG} which are smooth on $B(x_0,2R_0)\times (0,2T)$.  The extension of $\theta$ in the positive $z$-direction is denoted, \[\theta^*(x,z,t)=L(\theta(\cdot,t))(x,z).\]

Our formula for the time-integrated localized flux is obtained similarly to the local energy estimates derived in \cite{CaVa_QG} and \cite{CoWu_QG_09}. Multiplying \eqref{criticalSQG} by $\theta\phi_i$ where $\phi_i$ is a refined cut-off function associated with the ball $B(x_i,R)$ the extension of which in the positive $z$-direction is $\phi_i^*$, as well as using the properties of $L$, we infer,
\begin{align*}
F_i&=
\int_0^{2T}\int_{\R^2}\int_0^{\infty} \partial_z\big(\phi_i^* \theta^* z^b\partial_z \theta^*  \big)\,dz\,dx\,dt 
 - \int_0^{2T}\int  \frac 1 2\, {\theta^2} \, (\partial_t \phi_i)\,dx\,dt
\\&= \int_0^{2T}\int_{\R^2}\int_0^{\infty}  \phi_i^* |\nabla \theta^*|^2 z^b \,dz\,dx\,dt + \int_0^{2T}\int_{\R^2}\int_0^{\infty} \nabla \phi_i^* \cdot \nabla \theta^*  z^b \,\theta^* \,dz\,dx\,dt
\\&\quad - \int_0^{2T}\int  \frac 1 2\, {\theta^2} \, (\partial_t \phi_i)\,dx\,dt.
\end{align*}

\section{The temperature variance cascade}

Before stating the main result, we identify certain macro-scale averages taken over the time interval $[0,2T]$ and the ball $B(x_0,2R_0)$.  These are,
\begin{align*}
\vartheta_0&= \frac 1 T\int_0^{2T} \frac 1 {R_0^2} \int \frac 1 2 \theta^2 \phi_0^{2\delta-1}\,dx\,dt,
\\\vartheta_0^*&=\frac 1 T\int_0^{2T} \frac 1 {R_0^2} \int \bigg[ \int_0^{\infty} z^b \frac 1 2{\theta^*}^2 {\phi_0^*}^{2\delta-1}\,dz\bigg]\,dx\,dt,
\\\varepsilon_0^* &=\frac 1 T\int_0^{2T}    \frac 1 {R_0^2} \int \bigg[ \int_0^{\infty}z^b |\nabla\theta^*|^2 {\phi_0^*}\,dz\bigg]\,dx\,dt,
\end{align*}
and have dimensions $l^{2-2\alpha}$, $l^{4-3\alpha}$, and $l^{2-3\alpha}$ respectively.  These dimensions reflect the natural scaling associated with SQG -- namely, $\theta_l = l^{\alpha-1}\theta(lx,l^{\alpha}t)$ -- and are respected by the following length scale which determines the lower limit of physical scales over which the variance cascade will be shown to persist,
\[ \sigma_0=\max \bigg\{\bigg(\frac {\vartheta_0} {\varepsilon_0^*}\bigg)^{1/\alpha}, \bigg(\frac {\vartheta_0^*} {\varepsilon_0^*}\bigg)^{1/2} \bigg\}.
\]
 
The main result concerning the temperature variance cascade follows.

\begin{theorem}Fix $K_1$ and $K_2$ and
let $\beta=\min\{ (8C_0K_1K_2)^{-1/2} ,(8C_0K_1K_2)^{-1/\alpha}\}$ (note $C_0$ is fixed in the definition of our refined cut-off functions). If,
\[\sigma_0<\beta R_0,\]
then, for all $R$ satisfying $\frac 1 {\beta}\sigma_0\leq R\leq R_0$, any $(K_1,K_2)$-cover taken at scale $R$ yields an ensemble average satisfying,
\[\frac 1 {4K_1} \varepsilon_0^*\leq \langle F \rangle_R \leq 4K_2 \varepsilon_0^*. \] 
\end{theorem}
%The conclusion can be interpreted as saying that the ensemble averages of the (purely physical) temperature variance flux associated with any two covers are positive and mutually comparable. It is worth mentioning that the quasi-local nature of the Poisson-like kernel emphasizes the local behavior of $\theta$ and its derivatives and, therefore, the condition triggering the cascade is consistent with the physical theories which require that the temperature variance be small in comparison to its dissipation rate.
\begin{proof}
Invoking the definition of $(K_1,K_2)$-covers and the properties of our refined cut-off functions (the fixed constant $C_0$ is, in particular, introduced in that context), several relationships are immediate regarding the terms in the formula for $F_i$ and the associated ensemble averages.  First,
\begin{align*}
A_1:=\bigg|\frac 1 n \sum_{i=1}^n \frac 1 T \int_0^{2T}\frac 1 {R^2} \int_{\R^2} (\partial_t \phi_i)\,\frac {\theta^2} 2\,dx\,dt \bigg|&\leq \frac {C_0K_2} {R^\alpha}\vartheta_0. 
\end{align*}

Second, using the properties of refined cut-off functions and Young's inequality, we see that,
\begin{align*}
\bigg|\int_0^{2T}\int_{\R^2}\int_0^{\infty} \nabla \phi_i^* \cdot \nabla \theta^*  z^b \,\theta^* \,dz\,dx\,dt\bigg|
&\leq \frac 1 2 \int_0^{2T}\int_{\R^2}\int_0^{\infty} \phi_i^*\, |\nabla \theta^*|^2\, z^b\,dz\,dx\,dt
\\&\quad+\frac {C} {R^2} \int_0^{2T}\int_{\R^2}\int_0^{\infty} \frac 1 2 {\theta^*}^2{\phi_i^*}^{2\delta-1}z^b\,dz\,dx\,dt. 
\end{align*} 
 
Ensemble averaging the last term above gives, 
\[A_2:=\frac C {R^2} \frac 1 n \sum_{i=1}^n \int_0^{2T}\int_{\R^2}\int_0^{\infty} \frac 1 2 {\theta^*}^2{\phi_i^*}^{2\delta-1}z^b\,dz\,dx\,dt\leq \frac {C_0K_2} {R^2}\vartheta_0^*.\]
 
As the involved integrands are positive, direct comparison and the definition of $(K_1,K_2)$-covers yields the following interpolative estimate,\begin{align*} \frac 1 {K_1} \varepsilon_0^* \leq  \frac 1 n \sum_{i=1}^n \frac 1 T\int_0^{2T}   \frac 1 {R^2} \int \int_0^{\infty} z^b\,|\nabla\theta^*|^2 \phi_i^*\,dz\,dx\,dt\leq K_2 \varepsilon_0^*.
\end{align*}
Returning to our formula for $F_i$, taking the ensemble average and using the above facts reveals that,
\begin{align*}
\langle F \rangle_R &\leq \frac 3 2  \frac 1 n \sum_{i=1}^n \frac 1 T\int_0^{2T}  \frac 1 {R^2} \int \int_0^{\infty}   z^b\,|\nabla\theta^*|^2 \phi_i^*\,dz\,dx\,dt +A_1+A_2
\\&\leq 2 K_2 \varepsilon_0^* +\frac {C_0K_2} {R^2}\vartheta_0^*+\frac {C_0K_2} {R^\alpha}\vartheta_0,
\end{align*}
and,
\begin{align*}
\langle F \rangle_R &\geq \frac 1 2  \frac 1 n \sum_{i=1}^n \frac 1 T\int_0^{2T}   \frac 1 {R^2} \int  \int_0^{\infty} z^b\,|\nabla\theta^*|^2 \phi_i^*\,dz\,dx\,dt -A_1-A_2
\\&\geq \frac 1 {2K_1}\varepsilon_0^* -\frac {C_0K_2} {R^2}\vartheta_0^*-\frac {C_0K_2} {R^\alpha}\vartheta_0.
\end{align*} Recalling the definition of $\beta$ and the assumption that $\sigma_0/\beta\leq R\leq R_0$, we conclude our proof as we have shown that,
\begin{align*}
\frac 1 {4K_1} \varepsilon_0^* &\leq \langle F\rangle_R \leq 4K_2 \varepsilon_0^*.
\end{align*}
\end{proof}

Scale-locality of the temperature variance flux, i.e. that the average flux at a scale $R$ lying in the inertial range is most strongly correlated with the average fluxes at nearby scales, follows directly from Theorem 2 and, moreover, this correlation propagates exponentially along the dyadic scale (this is expected in SQG turbulence where large scale strain plays less of a role than in hydrodynamic turbulence, cf. \cite{P_etal}). We adopt the notation,
\[G_i=\frac 1 T \int_0^{2T}\int \frac 1 2 \,\theta^2 \, u\cdot\nabla\phi_i\,dx\,dt= R^2F_i,\]
for the time averaged localized temperature variance flux into the ball $B(x_i,R)$ and
$\langle G \rangle_R$ for the associated ensemble average over a $(K_1,K_2)$-cover taken at scale $R$.  The following corollary affirms the scale-locality of the temperature variance flux across the range of scales on which the cascade persists.  The proof is omitted.
\begin{corollary}If the premises of Theorem 2 are satisfied and $\frac 1 \beta \sigma_0 \leq r,R\leq  R_0$, then,
\[\frac {1 } {16K_1K_2} \bigg(\frac {r}{R}\bigg)^2 \leq \frac {\langle G \rangle_r}{\langle G \rangle_R}\leq 16K_1K_2\bigg(\frac {r}{R}\bigg)^2.\]
\end{corollary}
To see the exponential propagation along the dyadic scale, observe that, if $r=2^kR$ for some integer $k$, it follows that,
\[\frac {1 } {16K_1K_2}2^{2k} \leq \frac {\langle G \rangle_r}{\langle G \rangle_R}\leq 16K_1K_2 2^{2k}.\]

\subsection*{Acknowledgments.}

Z.B. acknowledges the support of the \emph{Virginia Space Grant Consortium} via the Graduate Research Fellowship; Z.G. acknowledges the support of the \emph{Research Council of Norway} via the grant 213474/F20, and the \emph{National Science Foundation} via the grant DMS 1212023.

\bibliographystyle{plain}
\bibliography{references}
\end{document}